\documentclass[11pt,a4paper]{article}

\usepackage{a4wide}
\usepackage{multirow}
\usepackage{empheq,fancybox}
\usepackage{graphicx}
\usepackage{amsfonts}
\usepackage{subfigure}
\usepackage[figuresright]{rotating}
\usepackage{tabularx}
\usepackage[table]{xcolor}
\usepackage{pstricks,pstricks-add,pst-node,pst-plot,pst-coil}
\usepackage{algorithm}
\usepackage{algorithmic}
\usepackage{url}

\newcommand{\MIPorig}{\textsc{Ilp}_{\mathrm{cb}}}
\newcommand{\MIPalt}{\textsc{Ilp}_{\mathrm{cs}}}
\newcommand{\LPorig}{\textsc{Lp}_{\mathrm{cb}}}
\newcommand{\LPalt}{\textsc{Lp}_{\mathrm{cs}}}
\definecolor{lightgray}{gray}{.85}
\newcommand{\ccg}{\cellcolor{lightgray}}

\newtheorem{theorem}{Theorem}[section]
\newtheorem{lemma}[theorem]{Lemma}

\newenvironment{proof}[1][Proof]{\begin{trivlist}
\item[\hskip \labelsep {\bfseries #1}]}{\end{trivlist}}

\newcommand{\qed}{\nobreak \ifvmode \relax \else
      \ifdim\lastskip<1.5em \hskip-\lastskip
      \hskip1.5em plus0em minus0.5em \fi \nobreak
      \vrule height0.75em width0.5em depth0.25em\fi}
\begin{document}

\title{Computational Performance Evaluation of Two Integer Linear Programming Models for the Minimum Common String Partition Problem}

\author{Christian Blum$^{1,2}$ and G{\"u}nther R.~Raidl$^{3}$  \\
~\\
$^1$Department of Computer Science and Artificial Intelligence\\ 
University of the Basque Country UPV/EHU, San Sebastian, Spain \\
{\sf christian.blum@ehu.es}\\
~\\
$^2$IKERBASQUE\\
Basque Foundation for Science, Bilbao, Spain\\
~\\
$^3$Institute of Computer Graphics and Algorithms\\
Vienna University of Technology, Vienna, Austria \\
{\sf raidl@ads.tuwien.ac.at}}

\date{}

\maketitle

\begin{abstract}
In the minimum common string partition (MCSP) problem two related input
strings are given. ``Related'' refers to the property that both strings
consist of the same set of letters appearing the same number
of times in each of the two strings. The MCSP seeks a minimum
cardinality partitioning of one string into non-overlapping substrings
that is also a valid partitioning for the second string. This problem
has applications in bioinformatics e.g.\ in analyzing related DNA or
protein sequences.
For strings with lengths less than about 1000 letters, a previously
published integer linear programming (ILP) formulation yields, 
when solved with a state-of-the-art solver such as CPLEX,  satisfactory results. 
In this work, we propose a new, alternative ILP model that is compared 
to the former one. While a polyhedral study shows the linear programming
relaxations of the two models to be equally strong, a comprehensive
experimental comparison using real-world as well as artificially
created benchmark instances indicates substantial computational advantages
of the new formulation.
\end{abstract}

\section{Introduction}
\label{intro}

String problems related to DNA and/or protein sequences are abundant in
bioinformatics. Well-known examples include the longest common
subsequence problem and its variants~\cite{HsuDu84lcs,SmiWat81:jmb}, the
shortest common supersequence problem~\cite{Gal2012:plosone}, and string
consensus problems such as the \emph{far from most string} problem and
the \emph{close to most string} problem~\cite{MouBabMon12a,MenOliPar05}.
Many of these problems are strongly \emph{NP}-hard~\cite{GJ79} and also computationally
very challenging.

This work deals with a string problem which is known as the \emph{minimum common string partition} (MCSP) problem. The MCSP problem can technically be described as follows. Given are two \emph{related} input strings $s^1$ and $s^2$ which are both of length $n$ over a finite alphabet $\Sigma$. The term \emph{related} refers to the fact that each letter appears the same number of times in each of the two input strings. Note that being related implies that $s^1$ and $s^2$ have the same length. A valid solution to the MCSP problem is obtained by partitioning $s^1$ (resp.~$s^2$) into a set $P^1$ (resp.~$P^2$) of non-overlapping substrings such that $P^1 = P^2$. The optimization goal consists in finding a valid solution such that $|P^1| = |P^2|$ is minimal.

Consider the following example. Given are sequences $s^1 = \mbox{\bf
AGACTG}$ and $s^2 = \mbox{\bf ACTAGG}$. Obviously, $s^1$ and $s^2$ are
related because {\bf A} and ${\bf G}$ appear twice in both input
strings, while {\bf C} and {\bf T} appear once. A trivial valid solution
can be obtained by partitioning both strings into substrings of length
one, that is, $P^1 = P^2 = \{\mbox{\bf A}, \mbox{\bf A}, \mbox{\bf C},
\mbox{\bf T}, \mbox{\bf G}, \mbox{\bf G}\}$. The objective 
value of this solution is six. However, the optimal solution, with
objective value three, is $P^1 = P^2 = \{\mbox{\bf ACT}, \mbox{\bf AG}, \mbox{\bf G}\}$.

The MCSP problem has applications, for example, in the bioinformatics field. Chen et al.~\cite{CheEtAl05:journal} point out that the MCSP problem is closely related to the problem of sorting by reversals with duplicates, a key problem in genome rearrangement. 

\subsection{History of Research for the MCSP Problem}

The original definition of the MCSP problem by Chen et
al.~\cite{CheEtAl05:journal} was inspired by computational problems
arising in the context of genome rearrangement such as:
May a given DNA string possibly be obtained by reordering subsequences 
of another DNA string? 
In the meanwhile, the general version of the problem was shown to be
\emph{NP}-hard~\cite{Goldstein2004}. Other papers concerning problem
hardness consider problem variants such as, for example, the $k$-MCSP
problem in which each letter occurs at most $k$ times in each input
string. The 2-MCSP problem was shown to be APX-hard
in~\cite{Goldstein2004}. Jiang et al.~\cite{Jiang2010} proved that the
decision version of the MCSP$^c$ problem---where $c$ indicates the size
of the alphabet---is \emph{NP}-complete when $c \geq 2$.

A lot of research has been done concerning the approximability of the
problem. Cormode and Muthukrishnan~\cite{Cormode2002}, for example,
proposed an $O(\log n \log^*n)$-approximation for the \emph{edit distance
with moves} problem, which is a more general case of the MCSP problem.
Other approximation approaches were proposed
in~\cite{Shapira2002,Kolman2007}. Chrobak et al.~\cite{Chrobak2004}
studied a simple greedy approach for the MCSP problem, showing that the
approximation ratio concerning the 2-MCSP problem is 3, and for the
4-MCSP problem the approximation ratio is in $\Omega(\log n)$. In the case
of the general MCSP problem, the approximation ratio lies between
$\Omega(n^{0.43})$ and $O(n^{0.67})$, assuming that the input strings
use an alphabet of size $O(\log n)$. Later Kaplan and
Shafir~\cite{Kaplan2006} improved the lower bound to~$\Omega(n^{0.46})$.
Kolman proposed a modified version of the simple greedy algorithm with
an approximation ratio of $O(k^2)$ for the $k$-MCSP~\cite{Kolman2005}.
Recently, Goldstein and Lewenstein~\cite{Goldstein2011} proposed a
greedy algorithm for the MCSP problem that runs in $O(n)$ time.
He~\cite{He2007} introduced another a greedy algorithm with the aim of
obtaining better average results.

Damaschke~\cite{raey} was the first one to study the fixed-parameter
tractability (FPT) of the problem. Later, Jiang et al.~\cite{Jiang2010}
showed that both the $k$-MCSP and MCSP$^c$ problems admit FPT algorithms
when $k$ and $c$ are constant parameters. Fu et al.~\cite{BinFu2011}
proposed an $O(2^nn^{O(1)})$ time algorithm for the general case and an
$O(n(\log n)^2)$ time algorithm applicable under certain constraints. 

Finally, in recent years researchers have also focused on algorithms for
deriving high quality solutions in practical settings. Ferdous and Sohel
Rahman~\cite{Ferdous2013,FerRah14:arxiv}, for example, developed a
{${\cal MAX}$}-{${\cal MIN}$} {A}nt {S}ystem metaheuristic. Blum et
al.~\cite{BluEtAl14:hm} proposed a probabilistic tree search approach.
Both works applied their algorithm to a range of artificial and real DNA
instances from~\cite{Ferdous2013}. The first integer linear programming
(ILP) model, as well as a heuristic approach on the basis of the
proposed ILP model, was presented in~\cite{BluEtAl15:ejor}. The heuristic is a 2-phase approach which---in the first phase---aims at covering most of the input strings with few but long substrings, while---in the second phase---the so-far uncovered parts of the input strings are covered in the best way possible. 
Experimental results showed that for smaller problem instances with $n < 1000$ 
applying a solver such as
CPLEX\footnote{\url{http://www-01.ibm.com/software/commerce/optimization/cplex-optimizer}}
to the proposed ILP is currently state-of-the-art. 
For larger problem instances, runtimes are typically too high and 
best results are usually obtained by the heuristic from~\cite{BluEtAl15:ejor}. 

\subsection{Contribution of this Work}

In this paper we introduce an alternative ILP model for solving the MCSP
problem. We show that the LP-relaxations of both models are
equally strong from a theoretical point of view.
An extensive experimental comparison with the model
from~\cite{BluEtAl15:ejor} shows, however, that CPLEX is able to derive 
feasible integer solutions much faster with the new model.
Moreover, the results when given the same computation time as for solving the existing ILP model are significantly better. 

\subsection{Organization of the Paper}

The remainder is organized as follows. In Section~\ref{sec:mip}, the ILP model from~\cite{BluEtAl15:ejor} as well as the newly
proposed ILP model are described. A polyhedral comparison of the two
models is performed in Section~\ref{sec:comparison}. 
The experimental evaluation on problem
instances from the related literature as well as on newly generated
problem instances is provided in Section~\ref{sec:experiments}. Finally,
in Section~\ref{sec:conclusions} we draw conclusions and give an outlook
on future work.

\section{ILP Models for the MCSP}
\label{sec:mip}

In the following we first review the existing ILP model for solving the
MCSP as proposed in~\cite{BluEtAl15:ejor}. Subsequently, the new
alternative model is presented.

\subsection{Existing ILP Model}
\label{sec:puremip}

The existing ILP model from~\cite{BluEtAl15:ejor} is based on the notion of \emph{common blocks}. Therefore we will henceforth refer to this model as the \emph{common blocks model}. A common block $b_i$ of input strings $s^1$ and $s^2$ is a
triple $(t_i,k^1_i,k^2_i)$ where $t_i$ is a string which appears
as substring in $s^1$ at position $k^1_i$ and in $s^2$
at position $k^2_i$, with $k^1_i,k^2_i\in\{1,\ldots,n\}$. 
Let the length of a common block $b_i$ be its string's length, i.e., $|t_i|$.
Let us now consider the set $B = \{b_1, \ldots, b_m\}$ of all existing
common blocks of $s^1$ and $s^2$.
Any valid solution $\mathcal{S}$ to the MCSP problem can then be
expressed as a subset of $B$, i.e., $\mathcal{S} \subset
B$, such that:
\begin{enumerate}
  \item $\sum_{b_i \in \mathcal{S}} |t_i| = n$, that is, the sum of the
  lengths of the common blocks in $\mathcal{S}$ is equal to the length
  of the input strings.
  \item For any two common blocks $b_i, b_j \in \mathcal{S}$ it holds that their corresponding strings neither overlap in $s^1$ nor in $s^2$.
\end{enumerate}

The ILP uses for each common block $b_i \in B$ a binary variable $x_i$ indicating its selection in the solution. In other words, if $x_i = 1$, the corresponding common block $b_i$ is selected for the solution. On the other side, if $x_i = 0$, common block $b_i$ is not selected.

\begin{empheq}[box=\shadowbox*]{align}
  (\MIPorig)\quad \text{min} \quad & \sum_{i=1}^{m} x_i &
  \label{eqn:objorig}\\
  \text{s.t.} \quad & \sum_{i\in \{1,\ldots,m \mid k^1_i \le j <
  k^1_i+|t_i|\}} x_i    = 1 &  \qquad & \mbox{for } j=1,\ldots,n \label{eqn:const2} \\
  & \sum_{i\in \{1,\ldots,m \mid k^2_i \le j <
  k^2_i+|t_i|\}} x_i    = 1 &  \qquad & \mbox{for } j=1,\ldots,n \label{eqn:const3} \\
  & x_i \in \{0, 1\} & & \mbox{for } i=1,\ldots,m\nonumber
\end{empheq}

The objective function~\eqref{eqn:objorig} minimizes the number of selected common
blocks. Equations~(\ref{eqn:const2}) ensure that each position
$j=1,\ldots,n$ of string $s^1$ is covered by exactly one selected common block
and selected common blocks also do not overlap.
Equations~(\ref{eqn:const3}) ensure the same with respect to $s^2$. 
Note that equations (\ref{eqn:const2}) (and also (\ref{eqn:const3})) implicitly
guarantee that the sum of the lengths of the selected blocks is $n$ as
\begin{equation*}
\sum_{i=1}^m |t_i| \cdot x_i =
  \sum_{i=1}^m\ \sum_{j=k^1_i}^{k^1_i+|t_i|-1} x_j =
  \sum_{j=1}^n\ \sum_{i\in \{1,\ldots,m \mid k^1_i \le j <
  k^1_i+|t_i|\}} x_i = n. 
\end{equation*}
Finally, note that the number of variables in model $\MIPorig$ is of order $O(n^3)$.

\subsection{An Alternative ILP Model: The Common Substrings Model}

An aspect which the above model does not effectively exploit is the fact that,
frequently, some string appears multiple times at different positions
as substring in $s^1$ and/or $s^2$. 
For example, assume that string \textbf{AC} appears five times in $s^1$ and
four times in $s^2$. Model $\MIPorig$ will then
consider $5\cdot 4=20$ different common blocks, one for each pairing of
an occurrence in $s^1$ and in $s^2$. Especially when the cardinality of the
alphabet is low and $n$ large, it is likely that some smaller strings
appear very often and induce a huge set of possible common blocks $B$.
To overcome this disadvantage, we propose the following alternative
modeling approach.

Let $T$ denote the set of all (unique) strings that appear as substrings
at least once in both $s^1$ and $s^2$. For each $t \in T$, let $Q^1_{t}$
and $Q^2_t$ denote the set of all positions between $1$
and $n$ at which $t$ starts in input strings $s^1$ and $s^2$,
respectively.

We now use binary variables $y_{t,k}^1$ for each $t \in T$, $k \in Q^1_t$, and $y_{t,k}^2$ for each $t \in T$, $k \in Q^2_t$. In case $y_{t,k}^i = 1$, the occurance of string $t \in T$ at position $k \in Q^i_{t}$ in input string $s_i$ is selected for the solution (where $i \in \{1,2\}$). On the other side, if $y_{t,k}^i = 0$, the occurance of string $t \in T$ at position $k \in Q^i_{t}$ in input string $s_i$ is not selected. The new alternative model, henceforth also referred to as the \emph{common substrings model}, can then be expressed as follows.

\begin{empheq}[box=\shadowbox*]{align}
  (\MIPalt) \quad \text{min} \quad & \sum_{t \in T} \sum_{k \in Q^1_t} y_{t,k}^1
  \label{eqn:obj}\\
  \text{s.t.} \quad & \sum_{t \in T}\  \sum_{k \in Q^1_t \mid k\le j < k+|t|} 
  	y_{t,k}^1 = 1  & & \mbox{for } j=1,\ldots,n \label{eqn:const4} \\
  & \sum_{t \in T}\  \sum_{k \in Q^2_t \mid k\le j <
  k+|t|} y_{t,k}^2 = 1  & & \mbox{for } j=1,\ldots,n \label{eqn:const5} \\
  & \sum_{k \in Q^1_t} y_{t,k}^1 = \sum_{k \in Q^2_t} y_{t,k}^2  & &
	  \mbox{for } t \in T \label{eqn:const6} \\ 
  & y_{t,k}^1 \in \{0, 1\} & & \mbox{for } t \in T,\ k \in Q^1_t \nonumber \\ 
  & y_{t,k}^2 \in \{0, 1\} & & \mbox{for } t \in T,\ k \in Q^2_t \nonumber 
\end{empheq}

The objective function \eqref{eqn:obj} counts the number of chosen
substrings; note that $\sum_{t \in T} \sum_{k \in Q^2_t} y_{t,k}^2$
would yield the same value. Equations~(\ref{eqn:const4}) and
(\ref{eqn:const5}) ensure that for each position $j=1,\ldots,n$ of input
string $s^1$ (respectively, $s^2$) exactly one covering substring is
chosen. These equations consider for each position $j$ all substrings
$t\in T$ for which the starting position $k$ is at most $j$ and less
than $k+|t|$. Equations~(\ref{eqn:const6}) ensure that each
string $t\in T$ is chosen the same number of times within $s^1$ and
$s^2$. Similarly as in $\MIPorig$, the requirement that the sum of the
lengths of the selected substrings has to sum up to $n$ follows implicitly from
\eqref{eqn:const4} and \eqref{eqn:const5}.

Concerning the number of variables involved in model $\MIPalt$, the following can be observed. A string of length $n$ has exactly $n^2$ substrings of size greater than zero. In the worst case, input strings $s^1$ and $s^2$ are equal, which means that $2x^2$ variables are generated. Therefore, in the general case, the new model has $O(n^2)$ variables. 

\section{Polyhedral Comparison}
\label{sec:comparison}

We compare the two ILP models by projecting solutions of $\MIPorig$
expressed in terms of variables $x_i$, $i=1,\ldots,m$, into the space of
variables $y^1_{t,k}$, $t\in T,\ k\in Q^1_t$, and $y^2_{t,k}$, $t\in T,\
k\in Q^2_t$, from $\MIPalt$. A corresponding solution is obtained~by
\begin{equation}
y^1_{t,k}=\sum_{i\in\{1,\ldots,m \mid t_i=t \land k^1_i=k\}} x_i
\qquad\mbox{and}\qquad
y^2_{t,k}=\sum_{i\in\{1,\ldots,m \mid t_i=t \land k^2_i=k\}} x_i.
\label{eqn:yx}
\end{equation}
Let $\LPorig$ and $\LPalt$ be the linear programming relaxations of
models $\MIPorig$ and $\MIPalt$, respectively, obtained by relaxing the
integrality conditions. In the following we show
that both models describe the same polyhedron in the space of
$y$-variables and are thus equally strong from a theoretical point.

\begin{lemma}
The polyhedron defined by $\LPorig$ is contained in $\LPalt$.
\end{lemma}
\begin{proof}
We show that for any feasible solution to $\LPorig$, the solution in
terms of the $y$-variables obtained by \eqref{eqn:yx} is also feasible in
$\LPalt$.
For equations~\eqref{eqn:const4} replacing $y^1_{t,k}$ yields
\begin{equation}
\sum_{t \in T}\ \sum_{k \in Q^1_t \mid k\le j < k+|t|}\  
	\sum_{i\in\{1,\ldots,m \mid t_i=t \land k^1_i=k\}} x_i =
	\sum_{i\in\{1,\ldots,m \mid k^1_i \le j < k^1_i+|t^1_i|\}} x_i,
	\label{eqn:colcorr}
\end{equation}
which corresponds to the left side of \eqref{eqn:const2} and is thus always equal to one.
Equations~\eqref{eqn:const5} are correspondingly fulfilled.
For constraints~\eqref{eqn:const6} we obtain for each $t\in T$
\begin{equation*}
\sum_{k\in P^1_t}\ \sum_{i\in\{1,\ldots,m \mid t_i=t \land k^1_i=k\}} x_i = 
\sum_{i\in\{1,\ldots,m \mid t_i=t\}} x_i =
\sum_{k\in P^2_t}\ \sum_{i\in\{1,\ldots,m \mid t_i=t \land k^2_i=k\}} x_i,
\end{equation*}
and they are therefore also always fulfilled.
Last but not least, also $0\le y^1_{t,k} \le 1$ and $0\le y^2_{t,k} \le 1$
trivially hold due to \eqref{eqn:const2} and \eqref{eqn:const3}.
\end{proof}

\begin{lemma}
The polyhedron defined by $\LPalt$ is contained in $\LPorig$.
\end{lemma}
\begin{proof}
Due to the correspondence~\eqref{eqn:colcorr},
equations~\eqref{eqn:const2} can be written in terms of the
$y$-variables and therefore also hold for any feasible solution of $\LPalt$.
Correspondingly, equations~\eqref{eqn:const3} are always fulfilled for
any solution of $\LPalt$.
If one is interested in a specific solution in terms of the
$x$-variables for a feasible solution expressed by $y$-variables, 
it can be easily derived by considering
each $t\in T$ and assigning values to variables $x_i$ with
$i\in\{1,\ldots,m \mid t_i=t\}$ in an iterative, greedy fashion 
so that relations~\eqref{eqn:yx} are fulfilled for any $k^1_i$ and $k^2_i$. 
A feasible assignment of such values must always exist
as an individual $x_i$ variable exists for each possible pair 
of positions $Q^1_t$ in $s^1$ and positions $Q^2_t$ in $s^2$,
due to constraints~\eqref{eqn:const6}, and the variable domains.
\end{proof}

From the above results, we can directly conclude the following.
\begin{theorem}
$\LPorig$ corresponds to $\LPalt$ when projected into the domain of
$y$-variables, and therefore $\MIPorig$ and $\MIPalt$ yield the same
LP-values and are equally strong.
\end{theorem}

\section{Experimental Evaluation}
\label{sec:experiments}

Both $\MIPorig$ and $\MIPalt$ were implemented using GCC~4.7.3
and IBM ILOG CPLEX~V12.1. The experimental results were obtained on a
cluster of PCs with 2933~MHz Intel(R)~Xeon(R)~5670 CPUs having 12 nuclei 
and 32GB~RAM. Moreover, CPLEX was configured for single-threaded execution.

\subsection{Benchmark Instances}

Two different benchmark sets were used for the experimental evaluation.
The first one was introduced by Ferdous and Sohel Rahman
in~\cite{Ferdous2013} for the evaluation of their ant
colony optimization approach. This set contains in total 30 artificial
instances and 15 real-life instances consisting of DNA sequences, that is, $|\Sigma|=4$.
Remember, in this context, that each problem instance consists of two
related input strings. Moreover, the benchmark set consists of four
subsets of instances. The first subset (henceforth labelled
\textsc{Group1}) consists of 10 artificial instances in which the input
strings have lengths up to 200. The second subset (\textsc{Group2})
consists of 10 artificial instances with input string lengths in
$(200,400]$. In the third subset (\textsc{Group3}) the input strings of
the 10 artificial instances have lengths in $(400,600]$. Finally,
the fourth subset (\textsc{Real}) consists of 15 real-life instances of
various lengths in $[200,600]$. The second benchmark set that we used
is new. It consists of 10 uniformly randomly generated instances for each
combination of $n \in \{100, 200, \ldots, 1000\}$ and alphabet size $|\Sigma| \in \{4,
12, 20\}$. In total, this set thus consists of 300 benchmark instances.

\subsection{Results for the instances from~Ferdous and Sohel Rahman}

The results for the four subsets of instances from the benchmark set by
Ferdous and Sohel Rahman ~\cite{Ferdous2013} are shown in
Tables~\ref{tab:results:group1}-\ref{tab:results:real}, in terms of one
table per instance subset. The structure of these tables is as follows.
The first and second columns provide the instance identifiers and the input string length, respectively. Then the results of
$\MIPorig$ and $\MIPalt$ are shown by means of five columns each.
The first column provides the objective values of the best solutions
found within a limit of 3600~CPU seconds. In case optimality of the
corresponding solution was proven by CPLEX, the value is marked by an
asterisk. The second column provides computation times in the form X/Y, 
where X is the time at which CPLEX was able to find the first valid integer
solution, and Y the time at which CPLEX found the best
(possibly optimal) solution within the 3600s limit. The third 
column shows optimality gaps, which are the relative differences in
percent between the values of the best feasible solutions and the lower
bounds at the times of stopping the runs. The fourth column provides LP
gaps, i.e.,
the relative differences between the LP relaxation values and the
best (possibly optimal) integer solution values.\footnote{Note that we
confirmed, in this context, that in all cases the values of the LP
relaxations concerning $\MIPorig$ and $\MIPalt$ were equal.} Finally,
the last column lists the numbers of variables of the ILP models.
The best result for each problem instance
is marked by a grey background, and the last row of each table provides
averages over the whole table.

\begin{table}[t!]
\caption{Results for the 10 instances of \textsc{Group1}.}
\label{tab:results:group1}
\centering
\scalebox{0.8}{
\tabcolsep1.2ex
\begin{tabular}{rrrrrrrrrrrrrr} \hline\hline
{\bf id} & {\bf $n$} & $\;\;\;$ & \multicolumn{4}{c}{$\MIPorig$}                           & $\;\;\;$ & \multicolumn{4}{c}{$\MIPalt$} \\ \cline{4-8} \cline{10-14}
         &         & & {\bf value} & {\bf time (s)} & {\bf gap} & {\bf LP gap} & {\bf \# vars} &          & {\bf value} & {\bf time (s)} & {\bf gap} & {\bf LP gap} & {\bf \# vars} \\ \hline
1         & 114 && $^*$\ccg41      &       0/1       &       0.0\% & 3.3\%  & 4299  && $^*$\ccg41  &    0/0  &     0.0\% & 3.3\% & 781 \\
2         & 137 && $^*$\ccg47      &       1/2       &       0.0\% & 3.6\%  & 6211  && $^*$\ccg47  &    0/0  &     0.0\% & 3.6\% & 928	\\
3         & 158 && $^*$\ccg52      &       2/34      &       0.0\% & 5.7\%  & 8439  && $^*$\ccg52  &    0/14 &     0.0\% & 5.7\% & 1172	\\
4         & 113 && $^*$\ccg41      &       0/1       &       0.0\% & 2.0\%  & 4299  && $^*$\ccg41  &    0/1  &     0.0\% & 2.0\% & 736	\\
5         & 119 && $^*$\ccg40      &       1/1       &       0.0\% & 3.0\%  & 4718  && $^*$\ccg40  &    0/1  &     0.0\% & 3.0\% & 833	\\
6         & 115 && $^*$\ccg40      &       0/3       &       0.0\% & 4.2\%  & 4435  && $^*$\ccg40  &    0/1  &     0.0\% & 4.2\% & 765	\\
7         & 162 && $^*$\ccg55      &       2/38      &       0.0\% & 2.9\%  & 8687  && $^*$\ccg55  &    0/18 &     0.0\% & 2.9\% & 1159	\\
8         & 123 && $^*$\ccg43      &       1/2       &       0.0\% & 3.2\%  & 4995  && $^*$\ccg43  &    0/2  &     0.0\% & 3.2\% & 816	\\
9         & 118 && $^*$\ccg42      &       1/2       &       0.0\% & 3.7\%  & 4995  && $^*$\ccg42  &    0/1  &     0.0\% & 3.7\% & 767	\\
10        & 170 && $^*$\ccg54      &       1/51      &       0.0\% & 3.7\%  & 9699  && $^*$\ccg54  &    0/16 &     0.0\% & 3.7\% & 1254 \\ \hline
{\bf avg.}& && \ccg45.5        &       1/14      &       0.0\% & 3.5\%  & 6029.3  && \ccg45.5  &    0/5  &     0.0\% & 3.5\% & 921.1   \\
\hline\hline
\end{tabular}}
\end{table}

\begin{table}[t!]
\caption{Results for the 10 instances of \textsc{Group2}.}
\label{tab:results:group2}
\centering
\scalebox{0.8}{
\tabcolsep1.2ex
\begin{tabular}{rrrrrrrrrrrrrr} \hline\hline
{\bf id} & {\bf $n$}  & $\;\;\;$ & \multicolumn{4}{c}{$\MIPorig$}                           & $\;\;\;$ & \multicolumn{4}{c}{$\MIPalt$} \\ \cline{4-8} \cline{10-14}
         &          & & {\bf value} & {\bf time (s)} & {\bf gap} & {\bf LP gap} & {\bf \# vars} &          & {\bf value} & {\bf time (s)} & {\bf gap} & {\bf LP gap} & {\bf \# vars} \\ \hline
1         & 337 && \ccg98  &	50/2067	    &	2.9\% & 4.1\% & 37743   && \ccg98     &     1/1218 & 2.2\%  & 4.1\% &	2740   \\
2         & 376 && 106	   &	80/1046     &	7.5\% & 7.8\% & 47174   && \ccg103    &     1/2554 & 3.6\%  & 5.2\% &	3191   \\
3         & 334 && 97	   &	35/1220	    &	2.7\% & 3.7\% & 36979   && $^*$\ccg96 &   1/523    & 0.0\%  & 2.7\% &	2776   \\
4         & 351 && 102	   & 	48/891	    &	4.9\% & 5.6\% & 40960   && \ccg100    &    1/470   & 2.2\%  & 3.7\% &	2914   \\
5         & 398 && 116	   &	83/2703     &	6.7\% & 7.5\% & 52697   && \ccg114    &    1/903   & 4.5\%  & 5.9\% &	3291   \\
6         & 327 && \ccg93  &	39/1476	    &	5.6\% & 6.5\% & 35650   && 94         &    1/269   & 6.2\%  & 7.5\% &	2694   \\
7         & 303 && 88	   &	31/3107	    &	6.0\% & 7.7\% & 30839   && \ccg87     &    1/1358  & 4.2\%  & 6.7\% &	2494   \\
8         & 358 && \ccg104 &	61/3248     &	5.1\% & 6.2\% & 42668   && \ccg104    &    1/72    & 5.7\%  & 6.2\% &	2954   \\
9         & 360 && 104	   &	49/1563	    &	5.2\% & 6.1\% & 42998   && \ccg103    &    1/162   & 4.2\%  & 5.2\% &	2924   \\
10        & 306 && 89	   &	27/1397	    &	3.6\% & 4.9\% & 31169   && $^*$\ccg88 &  1/434     & 0.0\%  & 3.8\% &	2423   \\ \hline
{\bf avg.}& && 99.7	   &	50/1872     &   5.0\% & 6.0\% & 39887.7 && \ccg98.7   & 1/796      & 3.3\%  & 5.1\% &  2840.1  \\
\hline\hline
\end{tabular}}
\end{table}

\begin{table}[t!]
\caption{Results for the 10 instances of \textsc{Group3}.}
\label{tab:results:group3}
\centering
\scalebox{0.8}{
\tabcolsep1.0ex
\begin{tabular}{rrrrrrrrrrrrrr} \hline\hline
{\bf id} & {\bf $n$}  & $\;\;\;$ & \multicolumn{4}{c}{$\MIPorig$}                           & $\;\;\;$ & \multicolumn{4}{c}{$\MIPalt$} \\ \cline{4-8} \cline{10-14}
         &          & & {\bf value} & {\bf time (s)} & {\bf gap} & {\bf LP gap} & {\bf \# vars} &          & {\bf value} & {\bf time (s)} & {\bf gap} & {\bf LP gap} & {\bf \# vars} \\ \hline
1         & 577 && 155       &	 333/858  &	 7.5\%  & 7.7\%  & 110973   && \ccg154   &   2/1015  &  6.4\%  & 6.5\% & 5230 \\
2         & 556 && 155       &	 345/693  &	 7.7\%  & 7.7\%  & 102670   && \ccg152   &   2/972   &  5.3\%  & 5.9\% & 4849 \\
3         & 599 && 166       &	 462/2063 &	 8.5\%  & 8.6\%  & 119287   && \ccg160   &   2/643   &  4.8\%  & 5.2\% & 5339 \\
4         & 588 && \ccg159   &	 458/976  &	 6.9\%  & 7.1\%  & 114975   && \ccg159   &   2/1783  &  6.4\%  & 7.1\% & 5251 \\
5         & 547 &&     150   &	 279/682  &	 9.7\%  & 9.9\%  & 99775    && \ccg147   &   3/237   &  7.6\%  & 8.1\% & 4917 \\
6         & 517 &&     147   &	 239/573  &	 9.1\%  & 9.2\%  & 88839    && \ccg143   &   2/621   &  6.0\%  & 6.7\% & 4441 \\
7         & 535 &&     149   &	 253/620  &	 9.8\%  & 10.0\% & 95765    && \ccg145   &   2/1572  &  6.7\%  & 7.5\% & 4734 \\
8         & 542 &&     151   &	 312/3591 &	 6.7\%  & 6.9\%  & 97400    && \ccg149   &   1/1092  &  5.0\%  & 5.7\% & 4691 \\
9         & 559 &&     158   &	 352/1022 &	 10.9\% & 11.1\% & 104186   && \ccg148   &   2/3418  &  4.2\%  & 5.1\% & 5009 \\
10        & 543 &&     148   &	 343/1334 &	 9.1\%  & 9.5\%  & 98237    && \ccg145   &   2/3316  &  6.7\%  & 8.2\% & 4823 \\ \hline
{\bf avg.}& &&     153.8 &        338/1241 &      8.6\%  & 8.8\%  & 103211.0 && \ccg150.2 &   2/1467 &  5.9\%  & 6.6\% & 4928.4 \\
\hline\hline
\end{tabular}}
\end{table}

\begin{table}[t!]
\caption{Results for the 15 instances of set \textsc{Real}.}
\label{tab:results:real}
\centering
\scalebox{0.8}{
\tabcolsep1.2ex
\begin{tabular}{rrrrrrrrrrrrrr} \hline\hline
{\bf id} & {\bf $n$}  & $\;\;\;$ & \multicolumn{4}{c}{$\MIPorig$}                           & $\;\;\;$ & \multicolumn{4}{c}{$\MIPalt$} \\ \cline{4-8} \cline{10-14}
         &          & & {\bf value} & {\bf time (s)} & {\bf gap} & {\bf LP gap} & {\bf \# vars} &          & {\bf value} & {\bf time (s)} & {\bf gap} & {\bf LP gap} & {\bf \# vars} \\ \hline
1          & 252 && $^*$\ccg78 &	   14/968   & 0.0\% & 3.9\%  & 22799     &&  $^*$\ccg78  &    0/232   & 0.0\%   & 3.9\% & 1966  \\
2          & 487 &&     139    &	   196/441  & 9.2\% & 9.3\%  & 80523     &&      \ccg134 &    1/988   & 5.2\%   & 5.9\% & 4330  \\
3          & 363 &&     104    &	   61/3575  & 5.6\% & 6.4\%  & 45869     &&      \ccg102 &    1/115   & 3.9\%   & 4.6\% & 3052  \\
4          & 513 &&     144    &	   301/1353 & 6.5\% & 6.6\%  & 91663     &&      \ccg141 &    1/227   & 4.3\%   & 4.7\% & 4467  \\
5          & 559 &&     150    &	   379/1998 & 7.9\% & 8.2\%  & 108866    &&      \ccg148 &    2/3230  & 6.2\%   & 7.0\% & 5068 \\
6          & 451 &&     128    &	   170/3584 & 6.5\% & 7.0\%  & 70655     &&      \ccg124 &    1/1392  & 3.0\%   & 4.0\% & 3836  \\
7          & 458 &&     121    &	   180/1814 & 6.9\% & 7.6\%  & 73502     &&      \ccg119 &    1/2729  & 4.3\%   & 6.1\% & 4187  \\
8          & 433 &&     116    &	   127/3268 & 6.8\% & 7.6\%  & 65560     &&      \ccg115 &    1/607   & 5.5\%   & 6.8\% & 3879  \\
9          & 468 &&     131    &	   191/358  & 8.8\% & 8.9\%  & 75833     &&      \ccg127 &    1/844   & 5.2\%   & 6.1\% & 4130  \\
10         & 450 &&     130    &	   144/3429 & 6.1\% & 6.7\%  & 69560     &&      \ccg127 &    1/2669  & 3.1\%   & 4.5\% & 3876  \\
11         & 400 &&     110    &	   114/3591 & 4.8\% & 5.6\%  & 56160     &&      \ccg109 &    1/2309  & 3.3\%   & 4.8\% & 3546  \\
12         & 449 &&     126    &	   178/651  & 9.8\% & 10.2\% & 70861     &&      \ccg122 &    1/562   & 6.3\%   & 7.2\% & 3981  \\
13         & 579 &&     157    &	   469/2236 & 7.1\% & 7.9\%  & 115810    &&      \ccg155 &    2/835   & 6.1\%   & 6.7\% & 5251 \\
14         & 458 &&     130    &	   161/3099 & 6.7\% & 7.2\%  & 73449     &&      \ccg129 &    1/581   & 5.5\%   & 6.5\% & 3905  \\
15         & 510 &&     139    &	   295/1430 & 7.7\% & 7.9\%  & 91060     &&      \ccg135 &    2/712   & 4.4\%   & 5.2\% & 4556  \\ \hline
{\bf avg.} & &&     126.9  &         212/2120 & 6.7\% & 7.4\%  & 74163.9   &&    \ccg124.3 &    1/1202  & 4.4\%   & 5.6\% &  4002.0 \\
\hline\hline
\end{tabular}}
\end{table}

The following observations can be made. First, apart from the instances
of \textsc{Group1} which are all solved with both models to optimality, the
results for subsets \textsc{Group2}, \textsc{Group3} and \textsc{Real}
are clearly in favor of model $\MIPalt$. Only in one out of 35 cases
(leaving \textsc{Group1} aside) a better result is obtained with
$\MIPorig$, and in further four cases the results obtained with
$\MIPalt$ are matched. In all remaining cases the solutions obtained
with $\MIPalt$ are better than those obtained with $\MIPorig$. This
observation is confirmed by a study of the optimality gaps. They are
significantly smaller for $\MIPalt$ than for $\MIPorig$. One of the main reasons for the superiority of model $\MIPalt$ over $\MIPorig$ is certainly the difference in the number of the variables. For the instance of \textsc{Group1}, $\MIPorig$ needs, on average, $\approx6.5$ times more variables than $\MIPalt$. This factor seems to grow with growing instance size. Concerning instances of \textsc{Group2}, $\MIPorig$ requires, on average, $\approx 14.0$ times more variables. The corresponding number for \textsc{Group3} is $\approx 20.9$. Another reason for the advantage of $\MIPalt$ over $\MIPorig$ is that symmetries are avoided. Finally, a
last observation concerns the computation times: the first feasible
integer solution is found for $\MIPalt$, on average, in about $0.7\%$ of
the time that is needed in the case of $\MIPorig$. 

\subsection{Results for the New Instance Set}

The results for the new set of problem instances are presented in
Table~\ref{tab:results:new}. Each line provides the results of both
$\MIPorig$ and $\MIPalt$ averaged over the 10 instances for a
combination between $n$ and $|\Sigma|$. The results are presented for
each ILP model by means of six table columns. The first five represent the
same information as was provided in the case of the first benchmark set.
An additional sixth column (with heading {\bf \# opt}) indicates for
each row how many (out of 10) instances were solved to optimality.
The additional last table column (with heading {\bf Impr.~in $\%$})
indicates the average improvement in solution quality of $\MIPalt$ over
$\MIPorig$. The results permit, basically, to draw the same conclusions
as in the case of the results for the instance set treated in the
previous subsection. The application of CPLEX to $\MIPalt$ outperforms
the application of CPLEX to $\MIPorig$ both in final solution quality
and in the computation time needed to find the first feasible integer
solution. These differences in results become more pronounced with
increasing input string length and with decreasing alphabet size. In the
case of $|\Sigma|=4$, for example, the solutions provided by $\MIPalt$
are on average $5.0\%$ better than those provided by $\MIPorig$. The
superiority of $\MIPalt$ over $\MIPorig$ is also indicated by the number
of instances that were solved to optimality: 160 out of 300 in the case
of $\MIPorig$, and 183 out of 300 in the case of $\MIPalt$.

\begin{sidewaystable}
\centering
\caption{Average results for the 300 instances of the newly generated benchmark set.}
\label{tab:results:new}
\scalebox{0.9}{
\tabcolsep1.2ex
\begin{tabular}{rr r rrrrr r rrrrrrrrr} \hline\hline
$n$ & $|\Sigma|$ & $\;\;\;$ & \multicolumn{5}{c}{$\MIPorig$} & $\;\;\;$
& \multicolumn{5}{c}{$\MIPalt$} & & & & {\bf Impr.} \\ \cline{4-9} \cline{11-16}
& & & {\bf value} & {\bf time (s)} & {\bf \# opt} & {\bf gap} & {\bf LP
gap} & {\bf \# vars} &          & {\bf value} & {\bf time (s)} & {\bf \#
opt} & {\bf gap} & {\bf LP gap} & {\bf \# vars} && {\bf in $\%$} \\ \hline
     & 4  &&  \ccg37.3  & 0/0	  &    10/10  & 0.0\%  & 2.8\%  & 3425.6     &&	  \ccg37.3  & 0/0  	 & 10/10 &  0.0\%  & 2.8\%  & 649.7  && 0.0\% \\
100  & 12 &&  \ccg68.5  & 0/0	  &    10/10  & 0.0\%  & 0.2\%  & 993.3      &&	  \ccg68.5  & 0/0  	 & 10/10 &  0.0\%  & 0.2\%  & 324.0  && 0.0\% \\
     & 20 &&  \ccg79.8  & 0/0	  &    10/10  & 0.0\%  & 0.0\%  & 622.4      &&	  \ccg79.8  & 0/0  	 & 10/10 &  0.0\%  & 0.0\%  & 264.2  && 0.0\% \\ \hline
     & 4  &&  \ccg63.5  & 3/101   &    10/10  & 0.0\%  & 3.5\%  & 13498.5    &&	  \ccg63.5  & 0/34 	 & 10/10 &  0.0\%  & 3.5\%  & 1473.8 && 0.0\% \\
200  & 12 &&  \ccg119.2 & 0/0	  &    10/10  & 0.0\%  & 0.5\%  & 3824.6     &&	  \ccg119.2 & 0/0  	 & 10/10 &  0.0\%  & 0.5\%  & 762.8  && 0.0\% \\
     & 20 &&  \ccg146.2 & 0/0	  &    10/10  & 0.0\%  & 0.0\%  & 2301.1     &&	  \ccg146.2 & 0/0  	 & 10/10 &  0.0\%  & 0.0\%  & 591.6  && 0.0\% \\ \hline
     & 4  &&  88.5  & 21/2358	  &    1/10   & 3.2\%  & 4.7\%  & 30398.5    &&	  \ccg88.1  & 0/448	 & 4/10  &  1.9\%  & 4.3\%  & 2412.5 && 0.5\% \\
300  & 12 &&  \ccg165.3 & 1/3	  &    10/10  & 0.0\%  & 0.8\%  & 8478.6     &&	  \ccg165.3 & 0/1  	 & 10/10 &  0.0\%  & 0.8\%  & 1249.1 && 0.0\% \\
     & 20 &&  \ccg206.7 & 0/0	  &    10/10  & 0.0\%  & 0.02\% & 5029.6     &&	  \ccg206.7 & 0/0  	 & 10/10 &  0.0\%  & 0.02\% & 967.0  && 0.0\% \\ \hline
     & 4  &&  115.5 & 89/2159	  &    0/10   & 6.7\%  & 7.2\%  & 53658.5    &&	  \ccg113.0 & 0/1277	 & 0/10  &  3.9\%  & 5.2\%  & 3369.8 && 2.2\% \\
400  & 12 &&  \ccg208.9 & 3/47	  &    10/10  & 0.0\%  & 0.9\%  & 14887.2    &&	  \ccg208.9 & 0/3  	 & 10/10 &  0.0\%  & 0.9\%  & 1742.1 && 0.0\% \\
     & 20 &&  \ccg261.5 & 1/1	  &    10/10  & 0.0\%  & 0.1\%  & 8932.0     &&	  \ccg261.5 & 0/0  	 & 10/10 &  0.0\%  & 0.1\%  & 1366.8 && 0.0\% \\ \hline
     & 4  &&  139.3 & 192/870	  &    0/10   & 9.1\%  & 9.3\%  & 84004.2    &&	  \ccg134.7 & 0/793	 & 0/10  &  5.5\%  & 6.2\%  & 4411.8 && 3.4\% \\
500  & 12 &&  \ccg249.0 & 10/328  &   10/10   & 0.0\%  & 0.9\%  & 23173.1    &&	  \ccg249.0 & 0/26 	 & 10/10 &  0.0\%  & 0.9\%  & 2266.2 && 0.0\% \\
     & 20 &&  \ccg312.2 & 4/4	  &    10/10  & 0.0\%  & 0.2\%  & 13761.0    &&	  \ccg312.2 & 0/0  	 & 10/10 &  0.0\%  & 0.2\%  & 1803.3 && 0.0\% \\ \hline
     & 4  &&  162.2 & 487/1893	  &    0/10   & 9.4\%  & 9.5\%  & 120795.1   &&	  \ccg159.0 & 2/2043	 & 0/10  &  7.0\%  & 7.9\%  & 5451.3 && 2.0\% \\
600  & 12 &&  291.0 & 32/1202	  &    2/10   & 0.9\%  & 1.2\%  & 33372.6    &&	  \ccg290.5 & 0/151	 & 9/10  &  0.1\%  & 1.1\%  & 2780.3 && 0.2\% \\
     & 20 &&  \ccg362.3 & 6/12	  &    10/10  & 0.0\%  & 0.3\%  & 19543.8    &&	  \ccg362.3 & 0/1  	 & 10/10 &  0.0\%  & 0.3\%  & 2253.2 && 0.0\% \\ \hline
     & 4  &&  187.7 & 785/2856    &    0/10   & 10.0\% & 10.2\% & 164116.2   &&	  \ccg183.4 & 3/1680	 & 0/10  &  7.6\%  & 7.8\%  & 6459.3 && 2.3\% \\
700  & 12 &&  331.0 & 54/1811	  &    0/10   & 1.2\%  & 1.4\%  & 45303.9    &&	  \ccg330.2 & 1/962	 & 2/10  &  0.7\%  & 1.2\%  & 3312.0 && 0.2\% \\
     & 20 &&  \ccg408.9 & 12/120  &    10/10  & 0.0\%  & 0.4\%  & 26588.5    &&	 \ccg408.9  & 0/4        & 10/10 &  0.0\%  & 0.4\%  & 2729.3 && 0.0\% \\ \hline
     & 4  &&  221.6 & 1442/3432   &    0/10   & 14.7\% & 15.3\% & 213956.1   &&	  \ccg207.1 & 5/2052	 & 0/10  &  8.9\%  & 9.4\%  & 7555.9 && 7.0\% \\
800  & 12 &&  368.7 & 123/2460	  &    0/10   & 1.6\%  & 1.8\%  & 59026.8    &&	  \ccg367.6 & 1/943	 & 0/10  &  1.1\%  & 1.5\%  & 3871.0 && 0.3\% \\
     & 20 &&  \ccg456.1 & 33/669  &    10/10  & 0.0\%  & 0.5\%  & 34451.6    &&	\ccg456.1   & 0/14       & 10/10 &  0.0\%  & 0.5\%  & 3180.1 && 0.0\% \\ \hline
     & 4  &&  266.3 & 1880/2314   &    0/10   & 22.3\% & 22.5\% & 271158.3   && \ccg227.3   & 6/2607     & 0/10  &  8.9\%  & 9.4\%  & 8682.5 && 17.2\% \\
900  & 12 &&  408.5 & 178/2406	  &    0/10   & 2.2\%  & 2.3\%  & 74372.5    &&	  \ccg405.5 & 1/1350	 & 0/10  &  1.3\%  & 1.5\%  & 4440.8 && 0.7\% \\
     & 20 &&  501.5 & 50/1625	  &    6/10   & 0.2\%  & 0.6\%  & 43543.4    &&	  \ccg501.3 & 0/238	 & 10/10 &  0.0\%  & 0.6\%  & 3649.8 && 0.04\% \\ \hline
     & 4  &&  288.7 & 3253/3739   &    0/10   & 21.8\% & 22.1\% & 334125.1   &&	  \ccg250.5 & 9/1465     & 0/10  &  10.0\% & 10.0\% & 9825.4 && 15.2\% \\
1000 & 12 &&  449.2 & 306/3147	  &    0/10   & 2.9\%  & 2.9\%  & 91955.2    &&	  \ccg443.2 & 1/1324	 & 0/10  &  1.4\%  & 1.7\%  & 5017.2 && 1.4\% \\
     & 20 &&  546.9 & 89/2182	  &    1/10   & 0.5\%  & 0.7\%  & 53736.0    &&	  \ccg546.1 & 1/844	 & 8/10  &  0.1\%  & 0.6\%  & 4106.7 && 0.1\% \\
\hline\hline
\end{tabular}}
\end{sidewaystable}

In order to facilitate the study of the computation times at which the first integer solutions were found, these times are graphically shown for different values of $|\Sigma|$ in three different barplots in Figure~\ref{fig:firstsoltime}. The charts clearly show that the advantages of $\MIPalt$ over $\MIPorig$ are considerable. In fact, the numbers concerning $\MIPalt$ are so small (in comparison to the ones concerning $\MIPorig$) that the bars are not visible in these plots. Moreover, these advantages seem to grow with increasing alphabet size. This means that, even though the differences in solution quality are negligible when $|\Sigma| = 20$, the first integer solutions are found much faster in the case of $\MIPalt$. The average gap sizes concerning the quality of the best solutions found and the best lower bounds at the time of termination are plotted in the same way in the three charts of Figure~\ref{fig:gapsize}. These charts clearly show that, for all combinations of $n$ and $|\Sigma|$, the average gap is smaller in the case of $\MIPalt$. Finally, Figure~\ref{fig:variables} shows evolution of the number of variables needed by the two models for instances of different sizes. 

\begin{figure}[!t]
\centering
\subfigure[$|\Sigma| = 4$]{
\label{fig:firstsoltime:a}
\includegraphics[angle=270,width=0.47\textwidth]{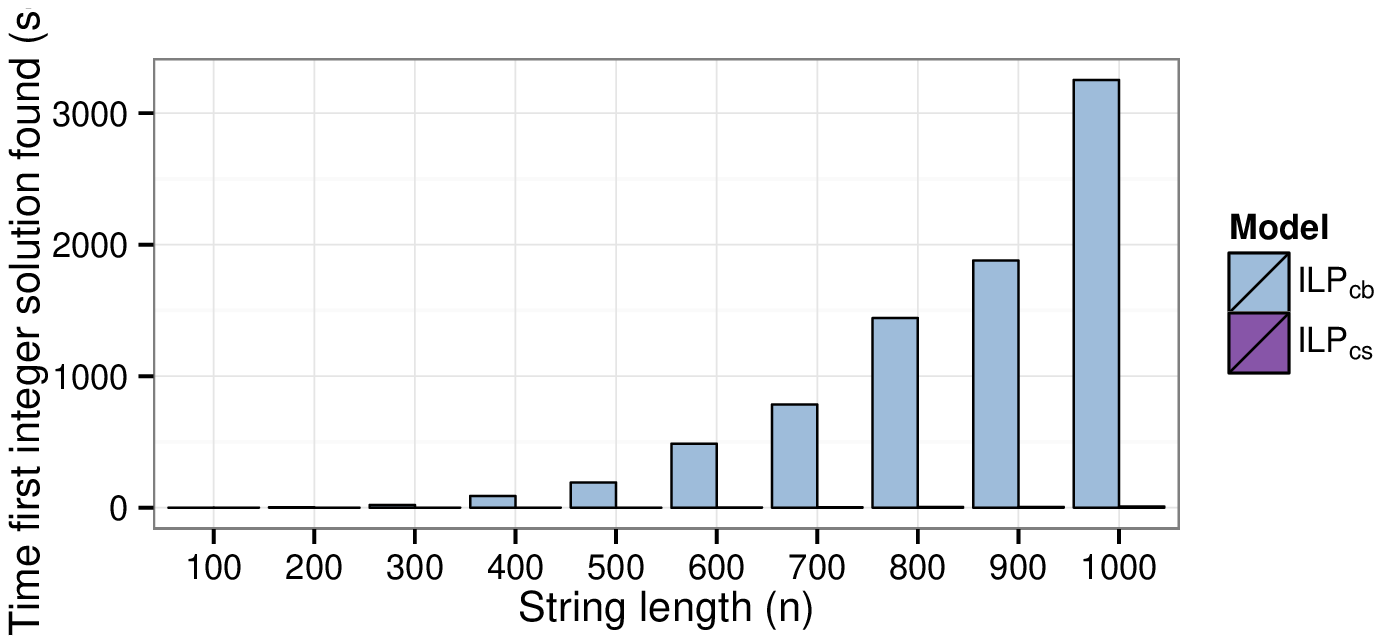}
}
\subfigure[$|\Sigma| = 12$]{
\label{fig:firstsoltime:b}
\includegraphics[angle=270,width=0.47\textwidth]{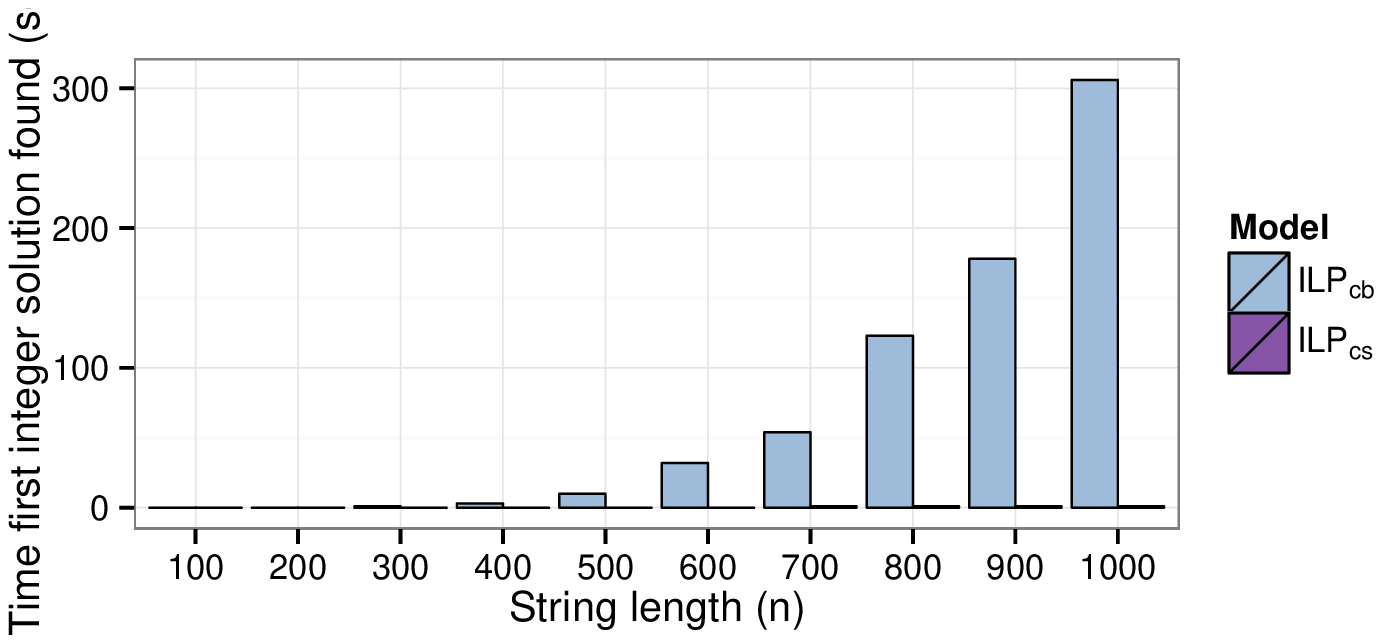}
}
\subfigure[$|\Sigma| = 20$]{
\label{fig:firstsoltime:c}
\includegraphics[angle=270,width=0.47\textwidth]{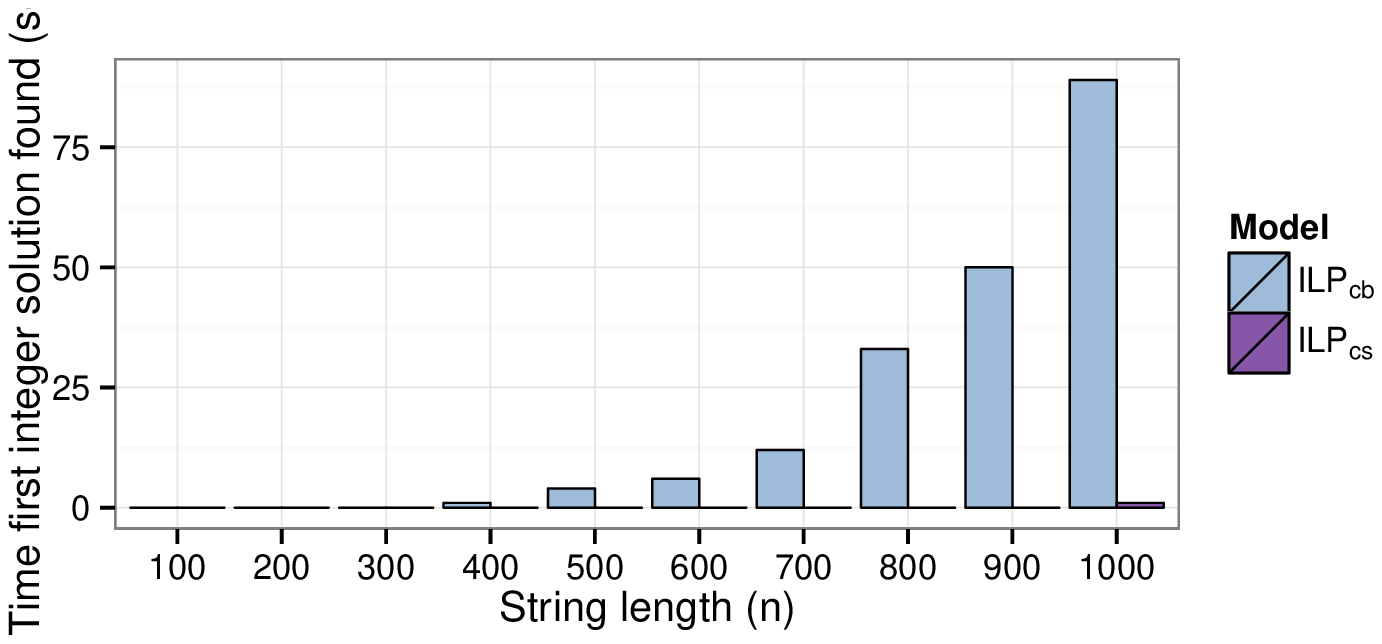}
}
\caption{Evolution of the average computation time the first integer solution is found.}
\label{fig:firstsoltime}
\end{figure}

\begin{figure}[!t]
\centering
\subfigure[$|\Sigma| = 4$]{
\label{fig:gapsize:a}
\includegraphics[angle=270,width=0.47\textwidth]{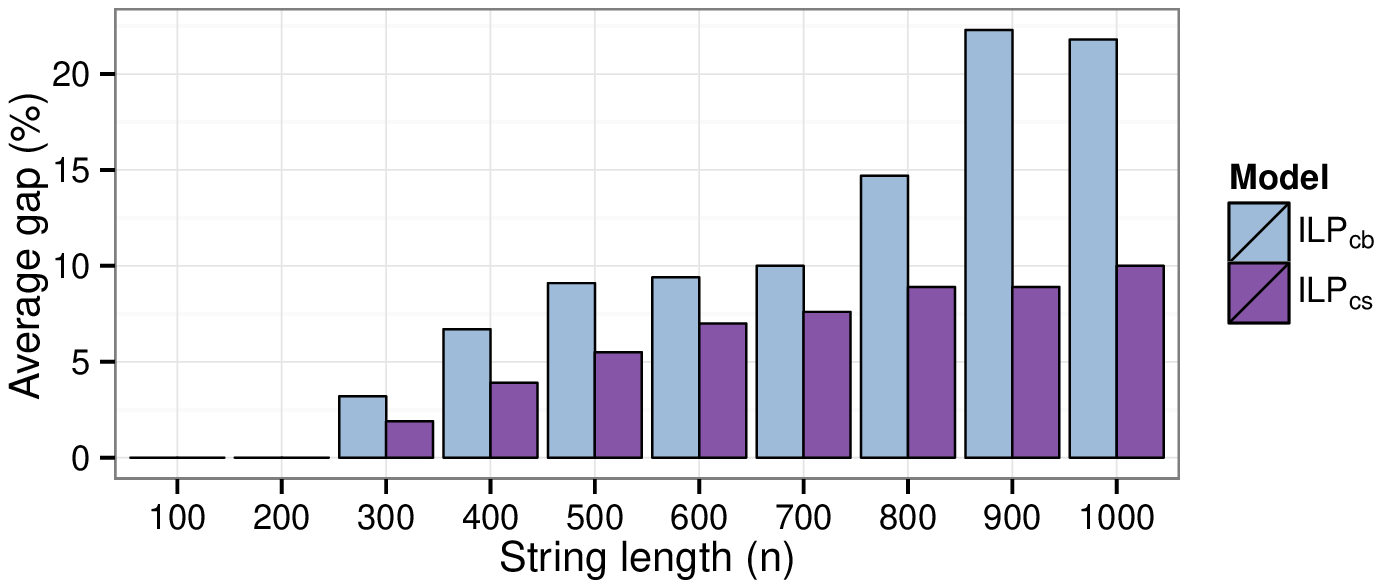}
}
\subfigure[$|\Sigma| = 12$]{
\label{fig:gapsize:b}
\includegraphics[angle=270,width=0.47\textwidth]{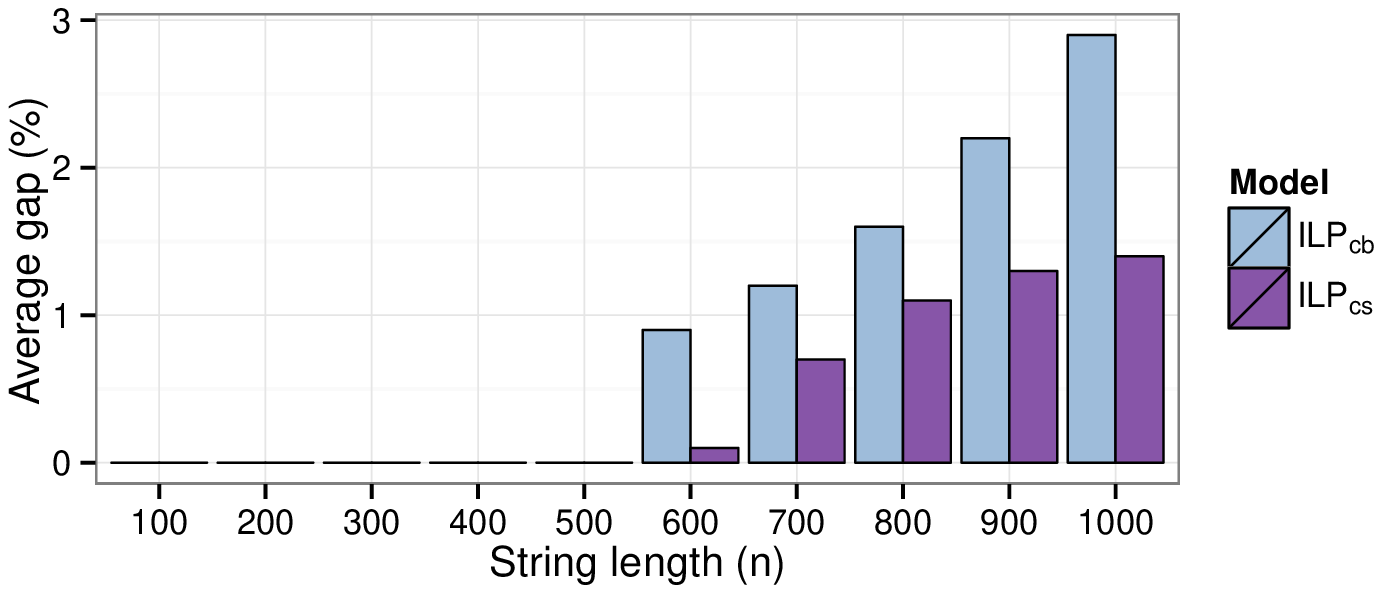}
}
\subfigure[$|\Sigma| = 20$]{
\label{fig:gapsize:c}
\includegraphics[angle=270,width=0.47\textwidth]{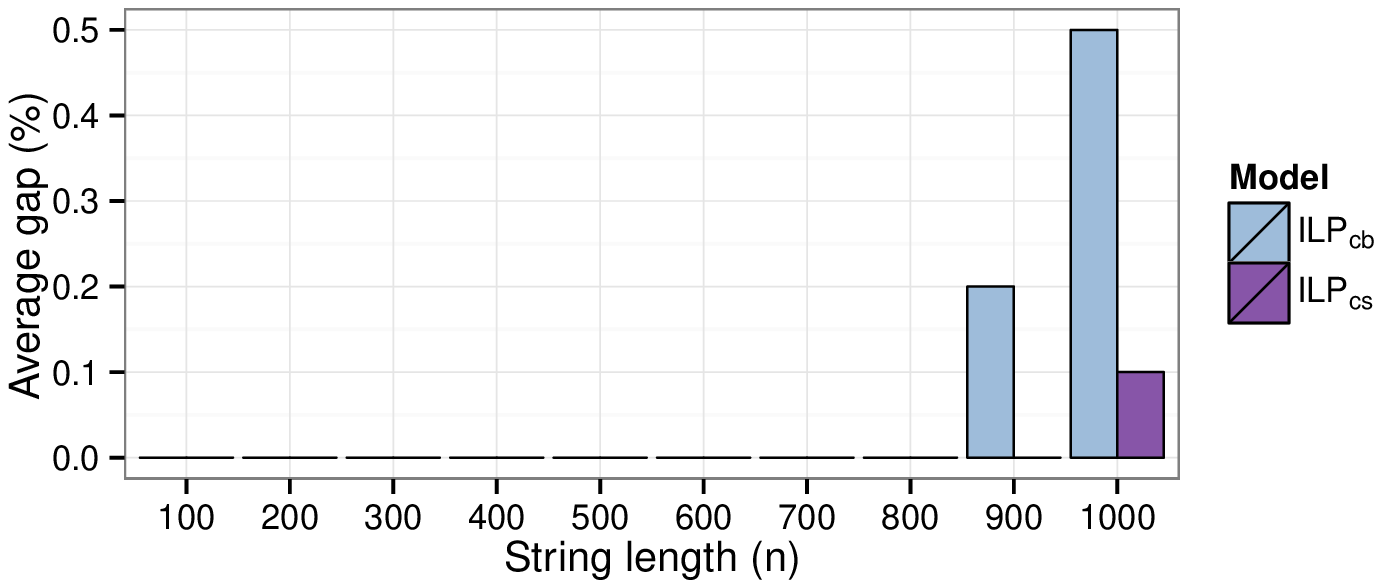}
}
\caption{Evolution of the average optimality gap size (in percent).}
\label{fig:gapsize}
\end{figure}

\begin{figure}[!t]
\centering
\subfigure[$|\Sigma| = 4$]{
\label{fig:variables:a}
\includegraphics[angle=270,width=0.47\textwidth]{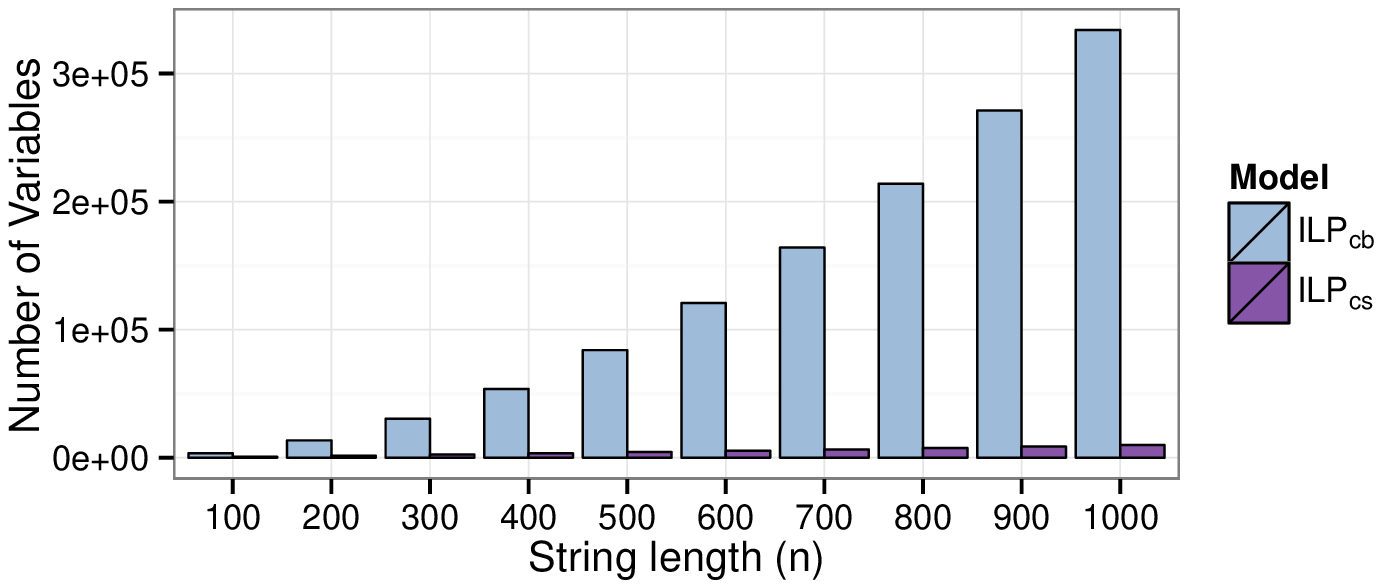}
}
\subfigure[$|\Sigma| = 12$]{
\label{fig:variables:b}
\includegraphics[angle=270,width=0.47\textwidth]{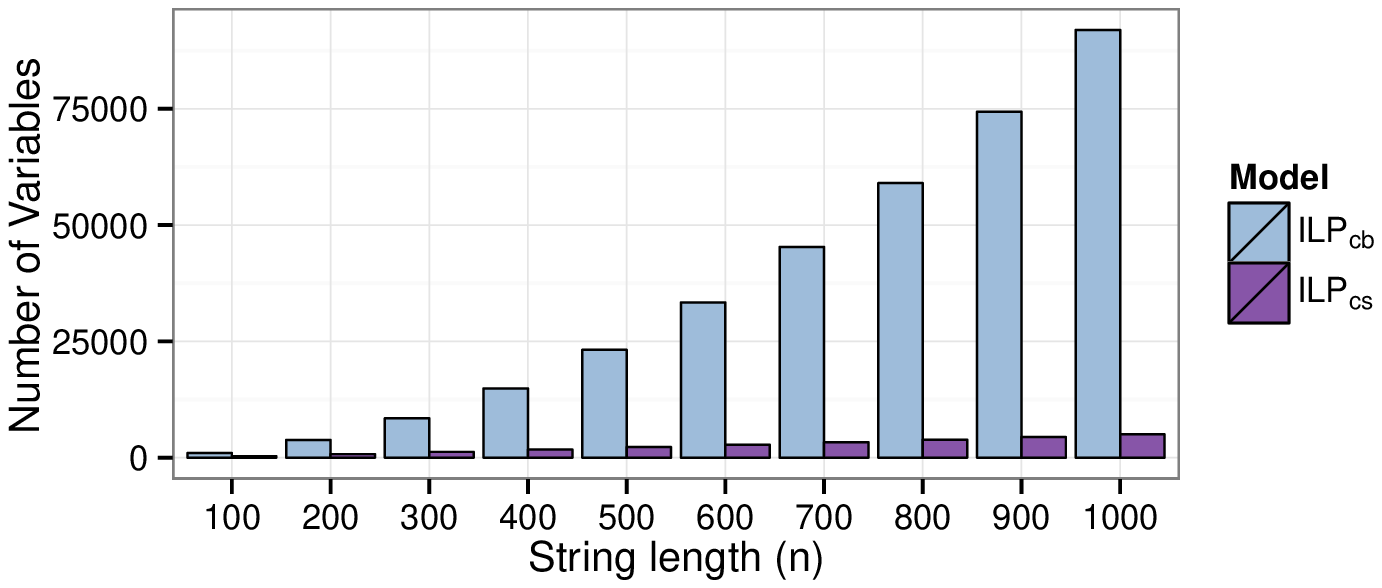}
}
\subfigure[$|\Sigma| = 20$]{
\label{fig:variables:c}
\includegraphics[angle=270,width=0.47\textwidth]{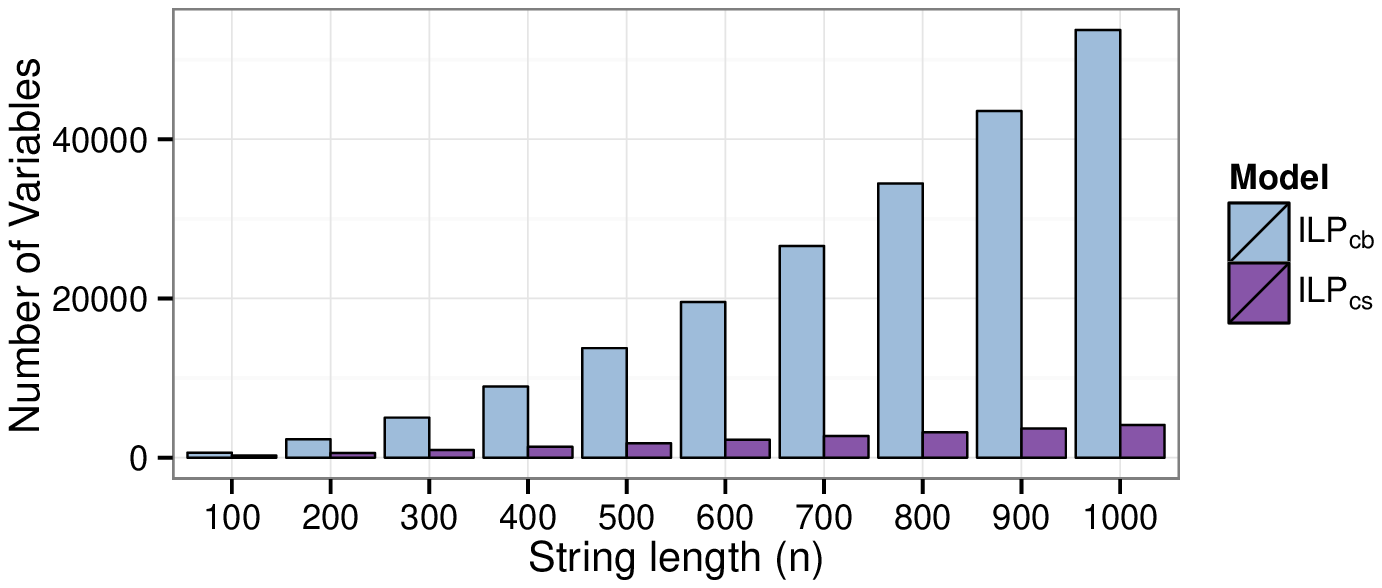}
}
\caption{Evolution of the number of variables used by the two ILP models.}
\label{fig:variables}
\end{figure}

\section{Conclusions and Future Work}
\label{sec:conclusions}

While (meta-)heuristic approaches are the state-of-the-art for
approximately solving large instances of the MCSP, instances with string
lengths of less than about 1000 letters can be well solved with 
an ILP model in conjunction with a state-of-the-art solver like CPLEX. 
In this work we have proposed the model based on \emph{common substrings} that
reduces symmetries appearing in the formerly suggested \emph{common blocks}
model. While our polyhedral analysis indicated that both models are
equally strong w.r.t.\ their linear programming relaxations, there are
significant differences in the computational difficulties to solve these
models. The new formulation allows for finding feasible solutions of
already reasonable quality in substantially less time and also yields 
better final solutions in most cases where proven optimal solutions could not 
be identified within the time limit. An important reason for this is to be found in the number of variables needed by the two models. While the existing model from the literature requires $O(n^3)$ variables (where $n$ is the length of the input strings), the new model only requires $O(n^2)$ variables.

In future work it would be interesting to consider
extended variants of the MCSP, in particular such where the input
strings need not to be related. In biological applications this would
give a greater flexibility as sequences that were also
affected by other kinds of mutations can be compared in terms of their 
reordering of subsequences. Another interesting generalization would be
to consider more than two input strings. The newly proposed ILP model 
appears to be a promising basis also for these variants.

\section*{Acknowledgements}
C.~Blum acknowledges support by grant TIN2012-37930-02 of the Spanish Government. In addition, support is acknowledged from IKERBASQUE (Basque Foundation for Science). Our experiments have been executed in the High Performance Computing environment managed by RDlab (\url{http://rdlab.lsi.upc.edu}) and we would like to thank them for their support.


\end{document}